\documentclass{jfm}
\usepackage{graphicx}
\usepackage{epstopdf, epsfig}
\usepackage{url} 
\usepackage[colorlinks, allcolors=blue]{hyperref}

\usepackage{natbib}
\usepackage[colorlinks]{hyperref}
\usepackage{etoolbox}

\makeatletter

\patchcmd{\NAT@citex}
  {\@citea\NAT@hyper@{%
     \NAT@nmfmt{\NAT@nm}%
     \hyper@natlinkbreak{\NAT@aysep\NAT@spacechar}{\@citeb\@extra@b@citeb}%
     \NAT@date}}
  {\@citea\NAT@nmfmt{\NAT@nm}%
   \NAT@aysep\NAT@spacechar\NAT@hyper@{\NAT@date}}{}{}

\patchcmd{\NAT@citex}
  {\@citea\NAT@hyper@{%
     \NAT@nmfmt{\NAT@nm}%
     \hyper@natlinkbreak{\NAT@spacechar\NAT@@open\if*#1*\else#1\NAT@spacechar\fi}%
       {\@citeb\@extra@b@citeb}%
     \NAT@date}}
  {\@citea\NAT@nmfmt{\NAT@nm}%
   \NAT@spacechar\NAT@@open\if*#1*\else#1\NAT@spacechar\fi\NAT@hyper@{\NAT@date}}
  {}{}

\makeatother

\newtheorem{corollary}{Corollary}
\newtheorem{theorem}{Theorem}

\shorttitle{Boundary Effects on Ideal Fluid Forces}
\shortauthor{K. I. McKee}

\title{Boundary Effects on Ideal Fluid Forces and Kelvin\textquotesingle s Minimum Energy Theorem}

\author{Kyle I. McKee\aff{1}
  \corresp{\email{kimckee@mit.edu}}
  }

\affiliation{\aff{1}Department of Mathematics, Massachusetts Institute of Technology,
Cambridge, MA 02139, USA
}

\begin{document}

\maketitle
\begin{abstract}
The electrostatic force on a charge above a neutral conductor is generally attractive. Surprisingly, that force becomes repulsive in certain geometries \cite[]{levin2011}, a result that follows from an energy theorem in electrostatics. Based on the analogous minimum energy theorem of \citet[]{kelvin1849}, valid in the theory of ideal fluids, we show corresponding effects on steady and unsteady fluid forces in the presence of boundaries. Two main results are presented regarding the unsteady force. First, the added mass is proven to always increase in the presence of boundaries. Second, in a model of a body approaching a boundary, where the unsteady force is typically repulsive \citep[\S 137]{Lamb1975}, we present a geometry where the force can be attractive. As for the steady force, there is one main result: in a model of a Bernoulli suction gripper, for which the steady force is typically attractive, we show that force becomes repulsive in some geometries. Both the unsteady and steady forces are shown to reverse sign when boundaries approximate flow streamlines, at energy minima predicted by Kelvin's theorem.
\end{abstract}

\section{Introduction}
It is generally stated that the electrostatic force on a charge above a conductor is attractive \citep[pp. 99]{Griffiths}.
However, \citet[]{levin2011} showed that the force may be repulsive in certain geometries (see figure \ref{fig:compare_electromag}C). Repulsion occurs when the electric field energy depends non-monotonically on the charge-conductor separation distance. Non-monotonicity occurs when conductor geometries resemble natural contours of electrical potential, as follows from the energy theorem of Thomson (Lord Kelvin) \citep[pp. 53]{jackson_classical}. An analogous energy theorem due to \citet[]{kelvin1849} holds in an ideal fluid. Whereas repulsion occurs in electrostatics when conductors resemble equipotential lines, we show surprising effects on fluid forces when boundaries resemble streamlines. In this paper, we analyse ideal fluid forces in such geometries, through consideration of both unsteady and steady components of force \citep[pp. 17]{sedov}.

The unsteady force on a submerged body is relevant during transient motions. It is often characterized by an effective added mass that limits large accelerations \citep[]{Newman2018, mckee2019} and  alters the natural vibration frequencies of submerged structures \citep{valentin2014,Newman2018}. Moreover, unsteady forces play a critical role in swimming mechanisms \citep[]{Saffman1967,Childress1981,Weymouth2013,Steele2017} of cephalopods and biomimetic robots \citep[]{Serchi_ScienceRobot}. 

The steady force is relevant even when there is no acceleration, and vanishes (along the direction of motion) in an infinite fluid according to d'Alembert's paradox \citep[404-405]{Batchelor2000}. The steady force is generally non-zero in higher connectivities; for example, there is a non-zero steady force between two interacting circular cylinders \citep[]{Wang2004,Tchieu2010}. Notably, the steady force is exploited for industrial applications in non-contact Bernoulli grippers \citep[]{giesen2013,davis2008}, in which a fluid source creates a region of high velocity and low pressure, resulting in an attractive lift force that can be used to manipulate objects (see \cite{vidref} for a video).

In the presence of boundaries, ideal fluid forces have been studied in various contexts. \citet[pp. 215]{Basset1888a} gave the exact added mass of a cylinder in a concentric confining cylinder. \citet[\S 93]{Lamb1975} solved the spherical version of that problem, which was verified experimentally by \cite{ackermann1964}. Confinement effects were studied further in the context of nuclear reactor core oscillations at
the Argonne National Laboratory\citep[]{Chung1984,Wambsganss1974}. \cite{brennen1982} then provided an extensive review of added mass and analysed boundary effects in various geometries. Later, \cite{Wang2004} presented analytical formulae for the fluid forces in the two-cylinder problem. 
More recently, \cite{Tchieu2010} formulated the two-body problem using a conformal map approach and revisited the two-cylinder problem.

In all these studies, boundaries were found to increase the added mass, relative to the boundary-free problem. One might wonder if there exists a boundary geometry in which the added mass of a body can actually be decreased. If this were possible, strategic boundary placement could give large accelerations during transient motions. We prove that no such boundary geometry exists, a result that follows from the minimum energy theorem of \cite{kelvin1849}. We then use Kelvin's minimum energy theorem to make a direct analogy with the electrostatic results of \cite{levin2011}. This analogy is discussed in detail in \S \ref{nonmon} and illustrated in figure \ref{fig:compare_electromag}.

Specifically, the unexpected repulsive forces found by \citet[]{levin2011} stem from a non-monotonic dependence of the electric field energy on the conductor-charge distance. If the energy is non-monotonic, then the force must take both positive and negative values, since force is related to the gradient of energy. Using an energy theorem of electrostatics, \citet[]{levin2011} showed that a non-monotonic energy results when conductors resemble contours of potential. By analogy, we pose the question: \emph{Does the ideal fluid kinetic energy of a system always depend monotonically on its separation distance to a boundary?} We answer this question in the negative and analyse the effect of energy non-monotonicity on ideal fluid forces. 
 
 The remainder of this paper is arranged as follows.  In \S \ref{kelvenerg}, we state Kelvin's minimum energy theorem along with an important corollary regarding boundaries. In \S \ref{nonmon}, we use the results of \S \ref{kelvenerg} to demonstrate that the fluid kinetic energy can be non-monotonic in the boundary separation distance, if boundaries approximate streamlines. \S \ref{unsteady} is dedicated to unsteady force effects. In \S \ref{addmass}, we show that boundaries cannot decrease the added mass. In \S \ref{nonmomunsteady}, we analyse the unsteady force by revisiting a calculation of \citet[\S 137]{Lamb1975}, but with a streamline-approximating boundary. In \S \ref{steady}, we analyse the steady force through two models of a Bernoulli gripper, which we solve exactly using the framework of \citet[]{crowdy2020solving}. In both \S \ref{nonmomunsteady} and \S \ref{steady}, forces reverse sign near energy minima predicted by Kelvin's minimum energy theorem, when boundaries approximate streamlines. 
 
 \section{Ideal Fluid Energy and Kelvin's Theorem}\label{energy}
\subsection{Kelvin's Minimum Energy Theorem}\label{kelvenerg}
 We begin by stating the minimum energy theorem of \cite{kelvin1849}, which will be generalized thereafter. The statement is as follows.
\begin{theorem}\label{thm1}
Consider an ideal fluid domain $D\subseteq \mathbb{R}^2$, with velocity field $\boldsymbol{v}(\boldsymbol{x})=\bnabla \phi(\boldsymbol{x})$ satisfying $\nabla^2\phi=0$ with no-flux boundary conditions on $\partial D$, where $\phi$ is single-valued. 
Further assume the volume of each boundary does not change in time and the velocity vanishes at infinity. Then any other incompressible flow $\boldsymbol{q}(\boldsymbol{x})$, satisfying the boundary conditions on $\partial D$, will possess kinetic energy greater than or equal to that of $\boldsymbol{v}$,
\begin{equation}\label{eq:theoremkelvin}
     \frac{1}{2} \rho \int_{D}||\boldsymbol{v}||^2 dV \leq \frac{1}{2} \rho \int_{D}||\boldsymbol{q}||^2 dV,
\end{equation}
when $||\boldsymbol{v}||^2$ contains no non-integrable singularities. \end{theorem}

When there exist source singularities in $D$, the integrals in (\ref{eq:theoremkelvin}) are not well-defined. This point is important since our model in \S \ref{steady} involves point sources. We reserve a discussion of this technicality for \S\ref{validitysource} and Appendix \ref{reg_proc}. In these analyses, we show that all results of \S \ref{energy} still hold if the flow possesses a source singularity, as long as the energy is defined appropriately. Until \S \ref{steady}, we will continue to assume that there are no non-integrable singularities in the flow. We now present the proof of Kelvin's minimum energy theorem as stated in Theorem \ref{thm1}, when the flow possesses no singularities.

\begin{proof}
We now prove Theorem \ref{thm1}. We choose to decompose $\boldsymbol{q}$ according to $\boldsymbol{q}=\boldsymbol{v}+\boldsymbol{w}$, where $\boldsymbol{w}$ is finite in $D$. Since $\boldsymbol{v}$ already satisfies the appropriate boundary conditions, $\boldsymbol{w}$ must satisfy the no-flux condition $\boldsymbol{w}\cdot\boldsymbol{n}=0$ on the boundary, $\boldsymbol{x}\in\partial D$. Noting that $(||\boldsymbol{q}||^2-||\boldsymbol{v}||^2)=2\bnabla \phi \cdot \boldsymbol{w}+||\boldsymbol{w}||^2$, we find that
\begin{equation}\label{eq:theoremkelvinsource2}
    \frac{1}{2} \rho \int_{D}\left(||\boldsymbol{q}||^2-||\boldsymbol{v}||^2\right) dV =    \rho\left(\int_{C}\phi \boldsymbol{w\cdot n}dA-\int_{D} \phi \bnabla \cdot \boldsymbol{w} dV\right)+\frac{1}{2} \rho \int_{D} ||\boldsymbol{w}||^2dV,
\end{equation}
where $C$ is equal to $\partial D$, since we took $\phi$ to be single valued. Then the first integral on the right hand side vanishes by the no-flux condition on $\boldsymbol{w}$. The second integral on the right hand side vanishes by the incompressibility of $\boldsymbol{q}$. Meanwhile, the last term is non-negative.
\end{proof}
In general, there is no unique solution to the Laplace problem with Neumann boundary conditions: solutions are only determined up to an unspecified amount of circulation around each boundary. In the above proof, assuming that $\phi$ was a single-valued function specified a particular solution to the Laplace problem. This solution is defined by the property that the flow $\nabla \phi$ possesses no circulation in the sense that $\int_{\gamma}\nabla \phi\cdot \boldsymbol{dl}=0$ around any closed loop $\gamma$. Hence, Theorem \ref{thm1} is applicable to the solution which possesses no circulation. As was noted by \citet{gonzalez2022variational}, the correct conclusion is that the minimal kinetic energy solution is the one possessing no circulation; in their paper, it is argued that the true solution is the one minimizing curvature. However, in the present paper, we focus on the zero circulation solution. 

In Appendix \ref{App_kelvCirc}, we present a generalization of Theorem \ref{thm1} to the case where arbitrary circulations may exist around each flow boundary. When circulation is present, the above proof is modified because the contour $C$ cannot be taken equal to $\partial D$; instead, $C$ must be chosen not to cross branch cuts associated with a multivalued $\phi$. It turns out that Theorem \ref{thm1} is still applicable even when there is circulation, provided that a particular flux condition is satisfied across branch cuts (see Appendix \ref{App_kelvCirc}). 

We emphasize again that in the present paper, we restrict our attention to the behaviour of the zero circulation solutions, to which Theorem \ref{thm1} applies. We proceed by showing that an added stationary boundary never decreases the energy. This corollary will be important in establishing the result of \S \ref{nonmon}, regarding energy non-monotonicity.

\begin{corollary}\label{corol}
Suppose a stationary impenetrable body $C$, with boundary $\partial C\subset D$, is added to the previous flow. Then a potential flow solution exists with velocity $\bnabla \phi_C$ defined over $D \backslash C$, where $\nabla^2 \phi_C=0$. Furthermore, the kinetic energy of that flow is greater than the boundary-less potential flow solution,
\begin{equation}\label{eq:coroll}
     \frac{1}{2} \rho \int_{D\backslash C} ||\bnabla \phi_C||^2 dV \geq \frac{1}{2} \rho \int_{D} ||\boldsymbol{v}||^2dV.
\end{equation}
\end{corollary}

\begin{proof}
The ideal flow, in the region $D \backslash C$, has a solution by standard existence theorems \citep[]{thomson1848theorems}. That solution has no flux into the region $C$ and a perfectly valid choice of an incompressible flow $\boldsymbol{q}$ defined over all $D$, as defined in Theorem \ref{thm2}, is
\begin{equation}\label{eq:corollareq}
  \boldsymbol{q}(\boldsymbol{x}) = \left\{
    \begin{array}{ll}
      \bnabla \phi_C & \boldsymbol{x}\in D\backslash C \\
      \boldsymbol{0} & \boldsymbol{x}\in C.
  \end{array} \right.
\end{equation}
This flow has precisely the kinetic energy of the potential flow solution in the presence of the boundary. However, by Theorem \ref{thm1}, the flow possesses more energy than the boundary-less flow. Note that $\partial C$ has measure zero, $\boldsymbol{v}$ is defined over $D$, and $\bnabla \phi_C$ over $D\backslash C$.
\end{proof}
In the proof of Theorem \ref{thm1}, it may have appeared that the reason $\boldsymbol{q}$ has more energy than $\boldsymbol{v}$ is either that the former is rotational or possesses non-zero circulation. However, we have demonstrated a third possibility via corollary \ref{corol}. Namely, if $\boldsymbol{q}$ is irrotational and possesses zero circulation, but is subject to additional boundary conditions, it will also possess an energy greater or equal to that of $\boldsymbol{v}$. This statement is understood most  simply by appealing to a variational argument; see, for example, \citet[pp. 434]{whitham1974linear} for a deeper discussion of the variational formulation of potential flow. The general argument is as follows. A solution to the Laplace equation $\phi$ can be viewed as a function that extremizes the energy functional $S[\phi(\boldsymbol{x})]=(\rho/2)\int{\left(\nabla\phi\cdot\nabla \phi\right)}dV$, so that $\delta S=0$, subject to a given set of constraints. In other words, $\phi(\boldsymbol{x})$ can be thought of as the result of a constrained optimization of $S[\phi(\boldsymbol{x})]$. In potential flow, the extremization occurs at a minimum of $S[\phi(\boldsymbol{x})]$ and the constraints are given by the boundary conditions of the flow. It follows that imposing additional boundary conditions (constraints) on the constrained optimization problem can never result in a better minimization of $S[\phi(\boldsymbol{x})]$.

This is in contrast to the electrostatic problem with a neutral conducting boundary. In the electrostatic problem, surface charge is allowed to distribute itself on the surface of the conductor, as long as the net charge is zero. The presence of surface charge affects the relevant energy functional. Hence, the addition of a conducting boundary presents an additional degree of freedom in the optimization of the electrostatic energy functional, and so neutral conducting boundaries decrease (or keep constant) the electrostatic energy.


\subsection{Non-Monotonic Kinetic Energy}\label{nonmon}
We now demonstrate that the fluid kinetic energy need not vary monotonically with the distance between a fluid system and a newly introduced boundary $\partial C$, a result that follows from corollary \ref{corol}. We define a fluid system by a collection of boundary conditions on $\partial D$ as in Theorem \ref{thm1}. A stationary boundary is introduced mathematically as an additional no-flux boundary $\partial C$, as in corollary \ref{corol}. As usual, slip is allowed in an ideal fluid. The boundary-separation distance is defined as the distance between a chosen point attached to $\partial C$ and one in $\partial D$. We now demonstrate by example, referring to figure \ref{fig:unsteadygeom}, as follows. 

Consider $D$ to be the exterior to a unit circle that translates vertically at unit velocity. We choose a boundary $\partial C$ comprising a pair of infinitely thin arcs resembling streamlines, which are drawn in three distance configurations, labelled $A,B,C$. The boundary-separation distance, labelled $x$, is taken as the vertical distance from the centre of the circle to the bottom of the boundary. We call the kinetic energy of the system $T(x)$. When $x\rightarrow \infty$, the boundary-less energy is recovered. At other distances, $T(x)\geq T(\infty)$ by corollary \ref{corol}. However, when $x=x_B$, the boundary does not perturb the cylinder flow, as it coincides precisely with streamlines, and $T(x_B)=T(\infty)$. Therefore, $T(x_B)$ corresponds to a local minima of the kinetic energy as a function of the boundary-separation distance. Hence, the kinetic energy dependence need not be monotonic. 

Similarly, one can conclude that the kinetic energy of the configuration in figure \ref{fig:compare_electromag}D is non-monotonic in the boundary separation distance, $y$. Since the perfect slits do not perturb the source flow when $y=0$, the kinetic energy $T(y)$ satisfies $T(0)=T(\infty)$. At other distances, $T(y)\geq T(\infty)$. One can then deduce the schematic plot of the kinetic energy given in $\ref{fig:compare_electromag}B$. A main goal of our investigation will be to analyse the physical manifestation of the change in sign of the slope of $T(y)$, in the highlighted region $y\gtrsim 0$.

The electrostatic energy of the configuration considered by \cite{levin2011} (see \ref{fig:compare_electromag}C) is plotted schematically in figure \ref{fig:compare_electromag}A. An energy theorem of electrostatics states that a neutral conducting boundary must decrease (or keep constant) the potential energy. The curve shape then follows from the fact that $V(0)=V(\infty)$ and $V(x)\leq V(\infty)$. Thus, the electrostatic energy theorem predicts the non-monotonic shape of the curve in figure \ref{fig:compare_electromag}A, where the highlighted window manifests as a repulsive electrostatic force, $F_{\mathrm{e}}=-dV/dx$. 

Although these discussions considered infinitely thin boundaries, we show in \S \ref{singlebody} that the corresponding effects on fluid forces persist for finite thicknesses.

\begin{figure}
  \centerline{\includegraphics[height=8.5cm]{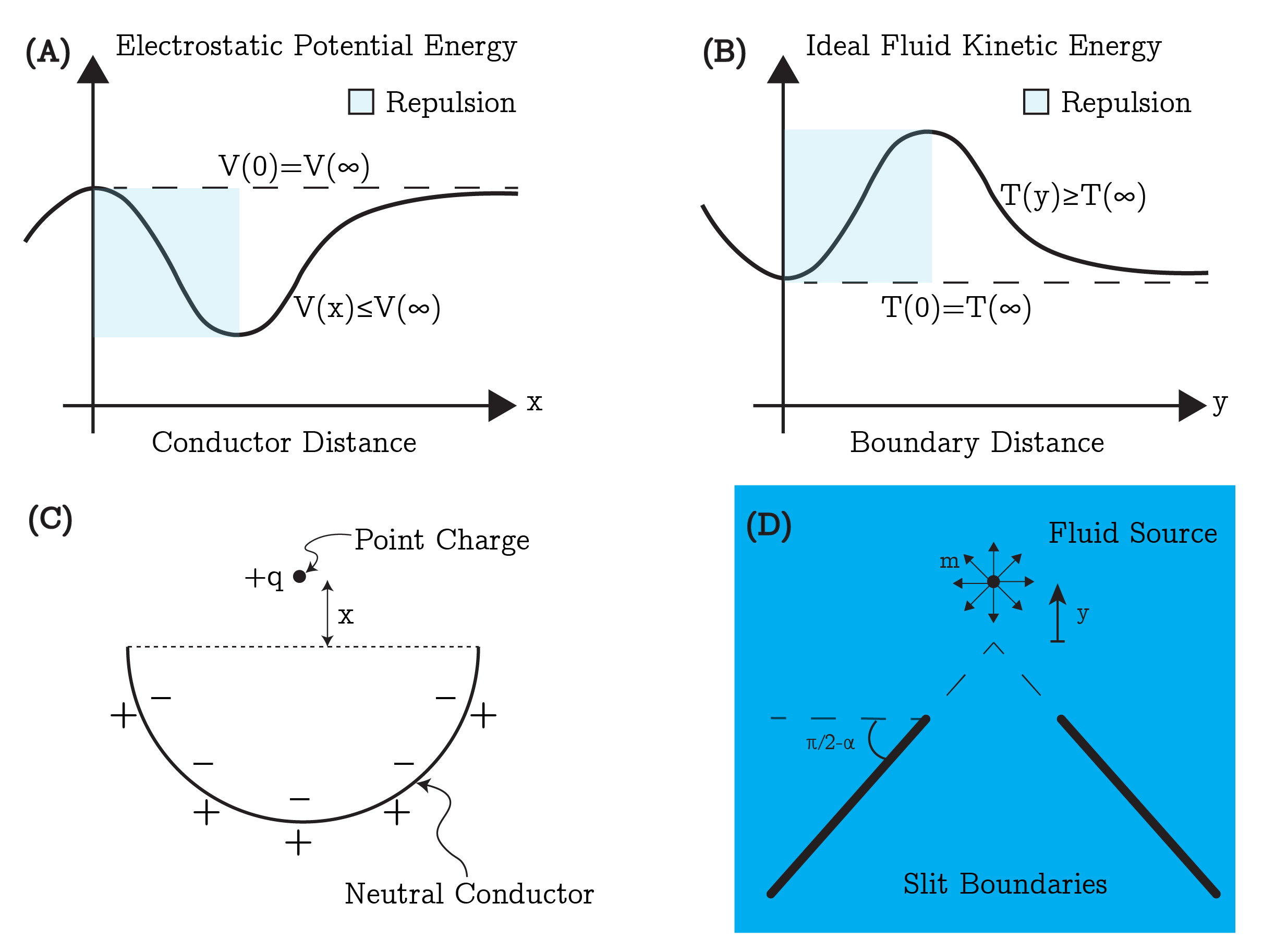}}
  \caption{\textbf{A)} Schematic illustration of the electrostatic energy of the geometry given in panel C.
  \textbf{B)} Schematic illustration of the ideal fluid kinetic energy of the geometry given in panel D.
  \textbf{C)} The electrostatic geometry considered by \cite{levin2011} is shown, where an electric point charge is placed above a semicircular neutral conducting arc. \textbf{D)} An ideal fluid geometry is given where, when $y=0$, the slits do not perturb the source flow.}
\label{fig:compare_electromag}
\end{figure}
 
\section{The Unsteady Force}\label{unsteady}
In \S\ref{addmass}, we show that the added mass is always greater than or equal to that in the boundary-less problem. An alternative derivation is given in Appendix \ref{altproof} using a conformal mapping approach, that does not appeal to Kelvin's energy Theorem. In \S\ref{nonmomunsteady}, we show the unsteady force found by \citet[\S 137]{Lamb1975} reverses sign when the boundary approximates a streamline. Herein, we ignore intrinsic body mass.

\subsection{Added Mass Cannot Decrease}\label{addmass}
Consider a single translating body in an unbounded ideal fluid. The Buckingham--Pi theorem \citep[]{buckingham1914} implies the total fluid  kinetic energy can be written as
\begin{equation}
    T_f=\frac{C_a}{2}\rho V U^2,
    \label{eq:kinetic}
\end{equation}
where $U$ is the body velocity, $C_a$ is a dimensionless constant called the added mass coefficient, and $\rho$ and $V$ are the fluid density and body volume, respectively. Considering acceleration at a constant rate $a$, (\ref{eq:kinetic}) gives the necessary work to move a distance $\delta x$, $\delta W=C_a \rho V a \delta x$. The resulting force is therefore $F=-\delta W/\delta x=-C_a \rho V a$; the reaction is equivalent to a mass augmentation of amount $C_a \rho V$. Corollary \ref{corol} showed that boundaries may not decrease the total kinetic energy $T_f$. Therefore, when the presence of a boundary is modelled as an effective change in the added mass, increasing $T_f$ implies an increased $C_a$ in (\ref{eq:kinetic}).

Hence, boundaries cannot be used to decrease the added mass and attain higher accelerations. The connection between boundary-induced added mass increase and Kelvin's theorem was mentioned by \citet[\S 93]{Lamb1975}. Although boundaries cannot reduce the added mass, geometries such as those described in \S \ref{nonmon} can lead to interesting fluid force effects, to be analysed in the remainder of this paper. The effects presented in the remainder of the paper rely on energy non-monotonicity as discussed in \S\ref{nonmon}, and are thus in analogy with the electrostatic results of \cite{levin2011}.
\begin{figure}
  \centerline{\includegraphics[height=5.5cm]{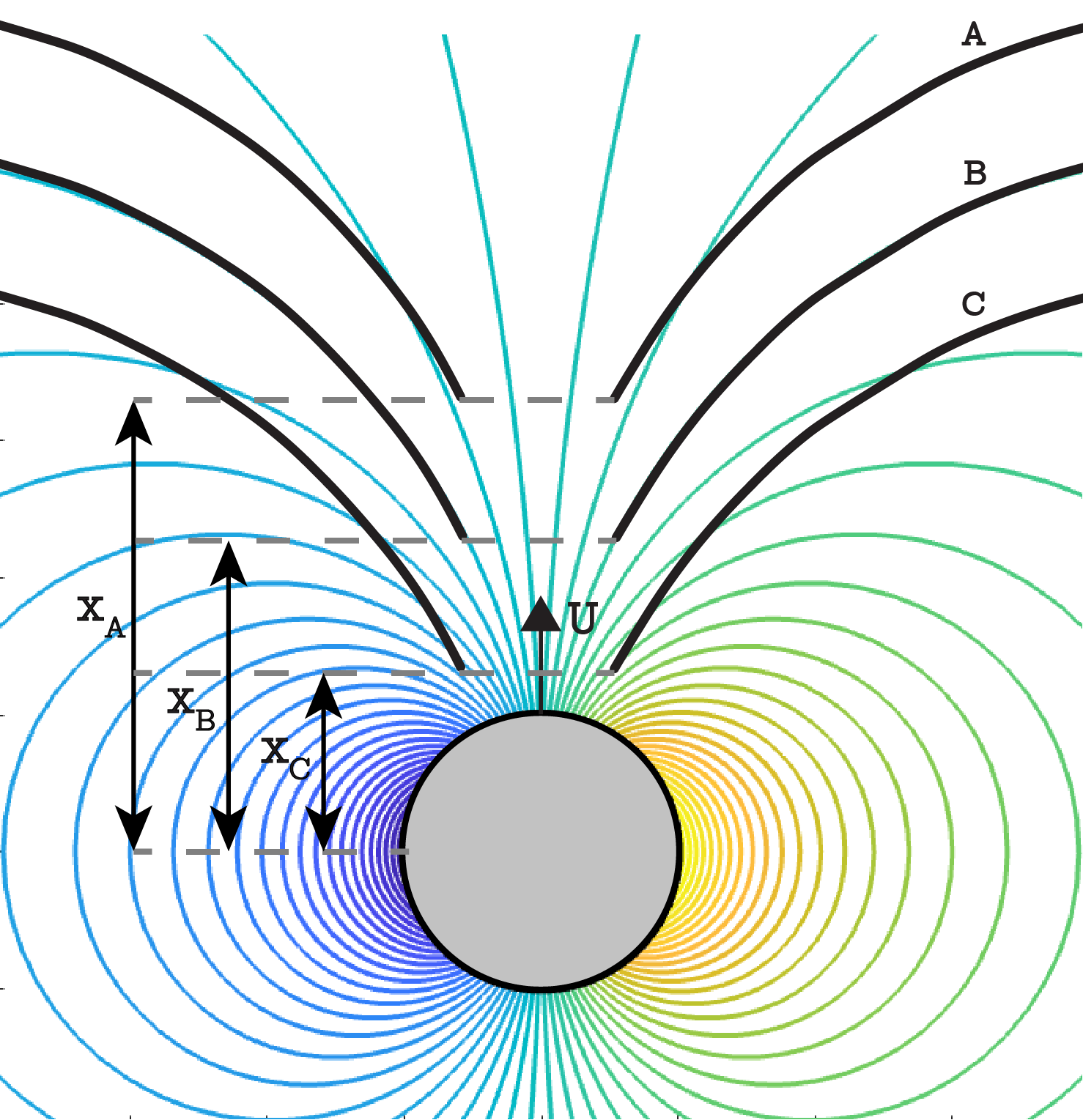}}
  \caption{A cylinder of unit radius and unit speed is pictured with its boundary-less streamlines. Configurations of a streamline-approximating boundary are superimposed and labelled $A,B,C$.}
\label{fig:unsteadygeom}
\end{figure}

\subsection{Effect of Non-monotonic Energy on Unsteady Force}\label{nonmomunsteady}
When a body translates toward a boundary, according to the added mass formulation, the energy can be expressed as in (\ref{eq:kinetic}) with $C_a=C_a(x)$, where $x$ is the separation distance to the boundary. Motion can only be taken as 1-dimensional after either assuming a constraint or by exploiting a symmetry in the $y$ direction, as is present in figure \ref{fig:unsteadygeom}. This assumption is made to simplify the calculation and illustrate our main argument. Note again that we will ignore the intrinsic mass of the cylinder. If the cylinder had mass $m_c$, then one should replace every $C_a$ with $C_a +m_c/(\rho V)$ in the following analysis.

For 1D free motion toward the boundary, the Euler-Lagrange equation is
\begin{equation}\label{eq:EOM}
    2\frac{\partial C_a(x)}{\partial x}\dot{x}^2+2C_a(x)\ddot{x}-\frac{\partial C_a(x)}{\partial x}\dot{x}^2=\frac{1}{\dot{x}}\frac{d}{dt}\left( C_a(x)\dot{x}^2 \right)=0,
\end{equation}
or $C_a(x)\dot{x}^2=\mathrm{const}$ since $\dot{x}\neq 0$. \citet[\S136]{Lamb1975} justified the Lagrangian formulation, and the validity of (\ref{eq:kinetic}), in the context of a body approaching a boundary with such a symmetry. In analysing the motion of a sphere toward a plane wall, \citet[\S 137]{Lamb1975} deduced from (\ref{eq:EOM}) a repulsive force since $\partial C_a/\partial x$ was always negative there. However, this need not always be the case as discussed in \S \ref{nonmon}: if the boundary coincides, at some finite $x$, with the streamlines of a traveling body then there is a local minima of $C_a(x)$. The equation of motion (\ref{eq:EOM}) implies, setting the constant to 1, that $\partial \dot{x}/\partial x=-C_a^{-\frac{3}{2}}(\partial C_a/\partial x)/2$ which switches sign at the local minima. 

Consider the geometry of figure \ref{fig:unsteadygeom}. A circle translates vertically and boundary-less streamlines, which are a function of velocity, are drawn for reference. The circle approaches a boundary that approximates flow streamlines. If, as the circle evolves according to (\ref{eq:EOM}), at some moment the free streamlines coincide with the boundary (as pictured for distance $x_B$), then the force will change sign when the boundary slits get aligned with the streamlines, $x=x_B$. In a different scenario where the circle is pushed at constant velocity, the necessary external force to maintain the velocity reverses sign at $x_B$.

\section{The Steady Force}\label{steady}
Here, we analyse the effect of energy non-monotonicity, as described in \S\ref{nonmon}, on the steady fluid force. We consider a source flow with streamline-approximating boundaries, as illustrated in figure \ref{fig:geometry}. In the simply connected problem (figure \ref{fig:geometry}A), attractive Bernoulli suction is achieved when the source is directly above the plate, giving a model of a Bernoulli suction gripper \citep[]{davis2008,giesen2013}. However, we show that the sign of the vertical force is reversed near the point where boundaries lie directly on streamlines. It is worth noting that in this 2-dimensional model, the geometry is not the typical Bernoulli gripper with axial symmetry as seen in \cite{vidref}; however the governing principle is the same. Both of the geometries of figure \ref{fig:geometry} are solved exactly; the simply connected problem is solved using conformal maps while the doubly connected problem is solved using the theoretical framework of \citet[]{crowdy2020solving}. We begin by demonstrating that the analysis of \S \ref{energy}, particularly corollary \ref{corol} and the non-monotonic energy discussion in \S \ref{nonmon}, is still applicable in the presence of a point source. 

\subsection{On the Validity of Kelvin's Theorem with Source Singularities}\label{validitysource}
In 2 dimensions, the kinetic energy of a source singularity, in an otherwise unbounded fluid, diverges logarithmically. A source at the origin has velocity $||\boldsymbol{v}|| \propto 1/r$, so that the absolute kinetic energy scales as $\int_0^{2\pi}\int_{\epsilon}^L ||\boldsymbol{v}||^2 r dr d\theta \propto \log{\left(L/\epsilon\right)}$. The integral is divergent both at the source location, as $\epsilon \rightarrow 0$, and at infinity, as $L\rightarrow \infty$. Therefore, the absolute kinetic energy of a source is not a meaningful quantity. However, forces and physical observables are related to differences, and in particular gradients, of energy. Therefore, the energy may be redefined with respect to any reference system, without altering physical predictions. In what follows, we will demonstrate a manner in which a finite energy can be defined even in the presence of sources. It is with respect to this definition of energy that the arguments of \S \ref{energy} will hold.

We note that a similar issue of energy divergence arises in the theory of electrostatics. The electromagnetic field energy of a point charge is non-integrable at the location of a charge in 3 dimensions. The divergence at the location of the source is eliminated (regularized) by choosing to measure the energy with respect to that of an isolated point source. Similarly, by choosing to measure the fluid kinetic energy with respect to that of an isolated source singularity, the divergence at the location of a fluid source can be regularized. What is left is to deal with the issue of the integral's convergence at infinity. The divergence at infinity is related to the fact that we are working in 2 dimensions. Sources of the Laplace equation are integrable at infinity in dimensions $n \geq 3$. 

The divergence at infinity may be resolved by two different approaches. The first approach involves treating the 2-dimensional problem as the limit of a 3-dimensional problem. The 2-dimensional flow due to a point source at $x=y=0$ is equivalent to the flow produced by a uniform continuum (line) of 3-dimensional sources along the entire $z$-axis. If one considers a line of sources with density $m$, in 3 dimensions, extending from $z=-L$ to $z=L$, the radial velocity profile in the $z=0$ plane is given by $v_r(r)=m/(2\pi r \sqrt{1+(r/L)^2})$; the integrated kinetic energy in the $z=0$ plane of such a source is then convergent at infinity for any finite $L$. In the presence of sources, the energy can be made finite at infinity by regulating each source according to the prescription $m/(2\pi r)\rightarrow m/(2\pi r \sqrt{1+(r/L)^2})$ for some finite regularizing parameter, $L$. 

A second approach, which does not require invoking 3 dimensions, is described by \citet[pp.125]{saffman1995vortex}. In this approach, a small circle of radius $\epsilon > 0$ around each singularity is excluded from the domain when computing the energy. The boundary at infinity is also replaced by a circle of large radius $R$. In the limit of $\epsilon\rightarrow 0$ and $R\rightarrow \infty$, the energy can be expressed as the sum of a finite part and a divergent constant $C \propto \log{\left(R/\epsilon\right)}$. Since constant shifts in energy are not observable, only the finite part of the energy needs to be retained. This procedure is reminiscent of the procedure of dimensional regularization used in quantum field theory, a connection nicely illustrated by \citet{olness2011regularization}.

The key takeaways from this section are as follows. The absolute kinetic energy of a source flow is not finite in 2 dimensions, relative to a static configuration, since the kinetic energy integral diverges both at the source location and at infinity. Measuring the energy relative to that of an isolated source eliminates the divergence at the source location. The divergence at infinity can be regulated in two ways, and a finite regularized energy can be defined. In terms of the regularized energy, all of the discussions of \S \ref{energy} apply, along with the interesting force behaviours predicted when boundaries approximate streamlines, even when there exists a source singularity in the flow. In what follows, we analyse the source flows of figure \ref{fig:geometry}, and examine the force behaviour when boundaries approximate the source streamlines.

\subsection{Single Slit and Source}\label{singlebody}
Consider an ellipse (which may degenerate to a slit) centered at the origin in the complex plane, with some orientation $\alpha$, as shown in figure \ref{fig:geometry}A. When $x=0$, the source is located directly above the plate and Bernoulli suction attracts the ellipse to the source, as shown in figure \ref{fig:plotsbernoulli}A. 
However, when $x \neq 0$ is fixed, there is a value of $y=x\tan{\alpha}$ where the ellipse falls on a natural source streamline. For a perfect slit,  $y=x\tan{\alpha}$ corresponds to an energy minima as a function $y$, as follows from the discussion of \S \ref{nonmon}. We thus expect the vertical steady force to flip sign near this point for sufficiently slender ellipses.

\begin{figure}
  \centerline{\includegraphics[height=5.5cm]{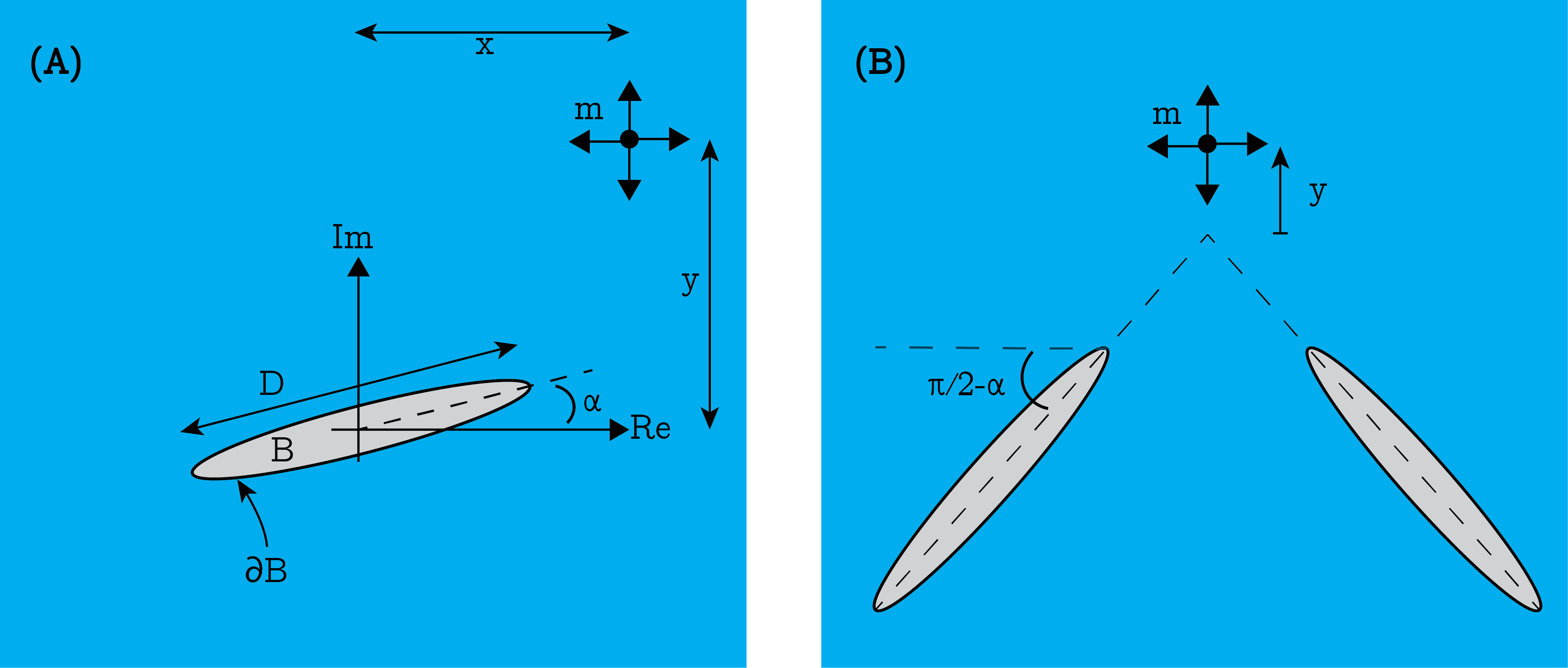}}
  \caption{\textbf{A)} Geometry of a source above a boundary. A source of strength $m$ is located at $x+\mathrm{i}y$ and a stationary ellipse centered on the origin and tilted at angle $\alpha$. When $\alpha=0$ and $x=0$, this is a 2D model of a Bernoulli gripper above a plate extending into the page.  \textbf{B)} Geometry of a source $m$ above a doubly-connected body, consisting of a pair of slits. When $\alpha=\pi/2$, this is a 2D model of a Bernoulli gripper above an object with a central hole.}
\label{fig:geometry}
\end{figure}
We employ a conformal mapping approach. The simpler problem of a source exterior to a unit disk is solved and mapped to the problem of interest (figure \ref{fig:geometry}A) by a conformal map. We denote the physical domain coordinates by $z$ and the pre-image coordinates by $\zeta$. The complex potential for a source of strength $m=2\pi$ at location $\zeta_s$, where $|\zeta_s|>1$, is $\log{\left(\zeta-\zeta_s\right)}$. The complex potential exterior to a stationary unit disc is then given by the Milne-Thomson theorem \citep[\S 6.21]{milne1962},
\begin{equation}\label{eq:thepotential}
    W=\log{\left(\zeta-\zeta_s\right)}+\log{\left((1-\overline{\zeta_s}\zeta)/\zeta\right)}.
\end{equation}
The unit disk exterior maps to the tilted ellipse exterior by a Joukowsky-type map,
\begin{equation}\label{eq:mapping}
    z(\zeta)=\left(\zeta + \frac{s}{\zeta}\right)e^{\mathrm{i}\alpha},
\end{equation}
where $s\in[0,1]$ for univalency \citep[]{Smith2008}; $s$ defines a homotopy between the circle ($s=0$) and the slit of length $D=2$ ($s=1$). Intermediate shapes are ellipses. To achieve the configuation in figure \ref{fig:geometry}A, $\alpha$ in (\ref{eq:mapping}) corresponds directly to $\alpha$ in the figure. However, the source pre-image, $\zeta_s$, must be chosen to ensure that $z(\zeta_s)\equiv z_s=x+\mathrm{i}y$. After some algebra, one finds the pre-image location,
$\zeta_s(z_s)=\left(z_s e^{-\mathrm{i}\alpha}+\sqrt{\left(z_s e^{-\mathrm{i}\alpha}\right)^2-4s}\right)/2$.

In what follows, we analyse the vertical component of force as $y$ is varied for a given $x$. The force is computed in the $\zeta$-plane since the complex potential is known there according to (\ref{eq:thepotential}). The steady force is given by the expression of \citet[equation 18]{Tchieu2010},
\begin{equation}\label{eq:forceexpress}
    F=\overline{\frac{\mathrm{i}}{2}\oint_{|\zeta|=1}{\left(\frac{dW}{d\zeta}\right)^2\left(\frac{dz}{d\zeta}\right)^{-1}d\zeta}},
\end{equation}
after setting the fluid density to be equal to 1. The vertical component is $F_s\equiv \Imag\left\{ F \right\}$. Using (\ref{eq:thepotential}) and (\ref{eq:mapping}), the residue theorem yields the vertical force on the ellipse,
\begin{equation}
    F_s=\Imag\left\{-\frac{2\pi\zeta_s^2(\zeta_s-s\overline{\zeta_s})}{(\zeta_s^2-s)^2(|\zeta_s|^2-1)}e^{-\mathrm{i}\alpha}\right\},
\end{equation}
which is plotted in figure \ref{fig:plotsbernoulli}A for the case of $x=0$ and $s=1$. The usual Bernoulli suction effect is captured there; a low pressure on the top of the plate leads to a net vertical (attractive) force. The suction persists regardless of the angle $\alpha<\pi/2$ when $x=0$. Note that taking $s<1$ does not qualitatively affect the shape of the curve.
 
We now examine the case where the slit may fall on a streamline of the source, $x\neq0$. Figure \ref{fig:plotsbernoulli}B plots the case of $x=2$ and $\alpha=\pi/4$ for various slenderness parameters, $s$. When $s=1$, the slit coincides with a natural (radial) streamline of the source at $y/D=1$. As per the discussion of \S \ref{nonmon}, this value of $y$ corresponds to a local energy minimum as a function of y, and the vertical force should thus change sign when $y/D=1$. The plot in figure \ref{fig:plotsbernoulli}B shows the force changes sign at precisely this point for a perfect slit, $s=1$. There is a local region of repulsion for $y/D>1$. For large $y/D$, the force becomes attractive again.  Meanwhile, the onset of sign-reversal of the force is shifted for an ellipse of finite slenderness ($s<1$) until the effect vanishes for $s=0.79$ where there is no repulsive region. For large $y/D$, all force plots are positive, indicating attraction.

\begin{figure}
  \centerline{\includegraphics[height=5.5cm]{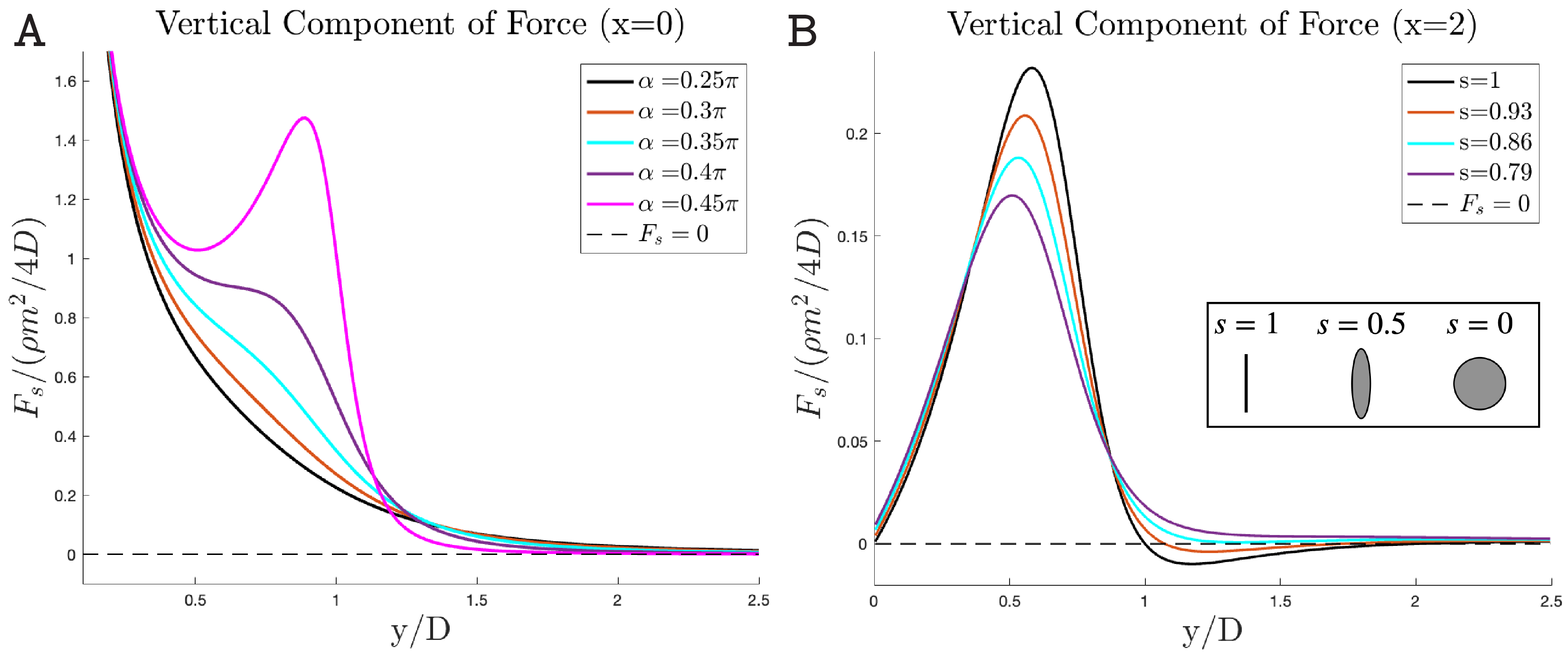}}
  \caption{\textbf{A)} Plot of the vertical steady force on the slit in the geometry shown in figure \ref{fig:geometry}A when $x=0$, $s=1$, for various $\alpha$. The force is positive for $y>0$, indicating attraction as in the standard Bernoulli gripper. \textbf{B)} Same plot when $\alpha=\pi/4$, $x=2$, for various slenderness parameters $s$. The slit  approximates a streamline when $y/D=1$, leading to an energy minima and local sign-reversal of the force only for slender ellipses ($s\approx 1$).}
\label{fig:plotsbernoulli}
\end{figure}

\subsection{Two Slits and Source}\label{twobody}
The vertical steady force is computed by a similar procedure in the two-body geometry shown in figure \ref{fig:geometry}B. The canonical domain for the doubly connected problem is the annulus, where the problem can be solved simply and then mapped to the physical domain by a conformal map (see figures \ref{fig:confmaps}A and \ref{fig:confmaps}D). The conformal map from the annulus, $|\zeta|\in [\rho,1]$, to the two-slit geometry shown in figure \ref{fig:geometry}B is
\begin{equation}\label{eq:2bodymap_43}
    z(\zeta)=-\mathrm{i}Ae^{3\mathrm{i}\alpha}\frac{\omega(\zeta,\sqrt{\rho}e^{2\mathrm{i}\alpha})\omega(\zeta,\sqrt{\rho}^{-1}e^{2\mathrm{i}\alpha})}{\omega(\zeta,\sqrt{\rho})\omega(\zeta,\sqrt{\rho}^{-1})},
\end{equation}
where $\omega$  is the Schottky--Klein prime function defined by \citet[pp. 64,75]{crowdy2020solving}; a similar map was used by \cite{crowdy2009spreading} to analyse the Weis-Fogh mechanism. Note that we consider figure \ref{fig:geometry}B with zero-thickness slits. To accommodate a source, a sink of equal strength is included in the annulus for mass conservation, with the sink located at the pre-image of infinity, $\zeta_{\infty}=\sqrt{\rho}$. The complex potential  in the $\zeta$-plane is \citep[]{sourcesink}
\begin{equation}
    W(\zeta)=\frac{m}{2\pi}\log{\left(\frac{\omega(\zeta,\zeta_s)\omega(\zeta,\overline{\zeta_s}^{-1})}{\omega(\zeta,\zeta_{\infty})\omega(\zeta,\overline{\zeta_{\infty}}^{-1})}\right)},
\end{equation}
up to a constant that does not enter our calculation. The source location in the annulus, $\zeta_s$, must be chosen such that its image lies on the imaginary axis as in figure \ref{fig:geometry}B. Conveniently, we notice the circle $|\zeta|=\sqrt{\rho}$ maps to the imaginary axis. We proceed by resolving the vertical component of force on the boundary. By symmetry, we can simply double the force on one slit; thus, the total force is twice that given in (\ref{eq:forceexpress}). Note that the map of (\ref{eq:2bodymap_43}) has zeros at the locations of the sharp slit edges. To avoid the corresponding poles while evaluating (\ref{eq:forceexpress}), the contour may be deformed slightly into the fluid domain, as was noted by \cite{Tchieu2010}.

The derivatives $dW/d\zeta$ and $dz/d\zeta$, and hence the integrand of (\ref{eq:forceexpress}), can be written analytically in terms of the functions $K(\zeta, \alpha)=\zeta\partial \log\omega(\zeta,\alpha)/\partial \zeta$. The integrand may then be explicitly evaluated using the series representations given by \citet[pp. 278,280]{crowdy2020solving}, and easily integrated numerically. The result is plotted in figure \ref{fig:plots2body} for the case of two slits near-parallel to the real axis. 

For large positive $y$, the force on the plates is positive, indicating attraction between the source and slits. For large negative $y$, the force on the plates is negative, indicating attraction again. However, a small region of repulsion is encountered near the point $y=0$, when the source is aligned with the slits. At small distances above the point $y=0$, when $y\gtrsim 0$, the source is evidently repelled by the slits as the force on the slits is negative. Just below the point $y=0$, when $y\lesssim 0$, the force on the slits becomes positive, indicating attraction again since the source is still above the slits. When $y$ becomes sufficiently negative, so that it lies below the slits, the force on the slits again becomes negative, indicating attraction as we already noted. 

All force plots vanish when $y=0$, and the force switches from attractive to repulsive. The location of this zero is precisely predicted by the discussion of \S\ref{nonmon}: when the natural source streamlines align with the slits, at $y=0$ as seen in figure \ref{fig:geometry}B, there exists a local energy minima as a function of $y$.

When the slits are parallel to the real axis, $\alpha=\pi/2$, the force is an odd function of $y$. Interestingly, the odd symmetry is quickly disrupted for deviations from $\alpha=\pi/2$. For $\alpha<\pi/2$, the repulsive region is retained, but the repulsive force magnitude is diminished. Since decreasing $\alpha$ pushes the reference point $y=0$ further away from the plates (see figure \ref{fig:geometry}B), it is to be expected that the magnitude of the repulsive force should decrease since the source flow decays with distance. Despite its diminishing magnitude, the repulsive region persists for $\alpha < \pi/2$, as is seen in figure \ref{fig:plots2body}.

As was found in the single body calculation of \S\ref{singlebody}, we expect the repulsion found in the two-body problem to be valid up to some finite slit widths. However, we do not present computations for finite slit widths in the two-body problem.

Additionally, we note that no Kutta condition was imposed in the calculation of \S \ref{twobody}. The forces in figure \ref{fig:plots2body} should thus be interpreted as those on slightly inflated slits with zero circulation and smooth corners. If the slits were truly perfect, a Kutta condition would be necessary, which would affect the vertical force. By slightly inflating the ellipse, the singularity in the conformal map is pushed off of the object boundary, into the exterior of the fluid domain; hence, this interpretation is consistent with our method of calculating (\ref{eq:forceexpress}) by pushing the contour slightly into the fluid domain to exclude the map singularity.

\begin{figure}
  \centerline{\includegraphics[height=5.5cm]{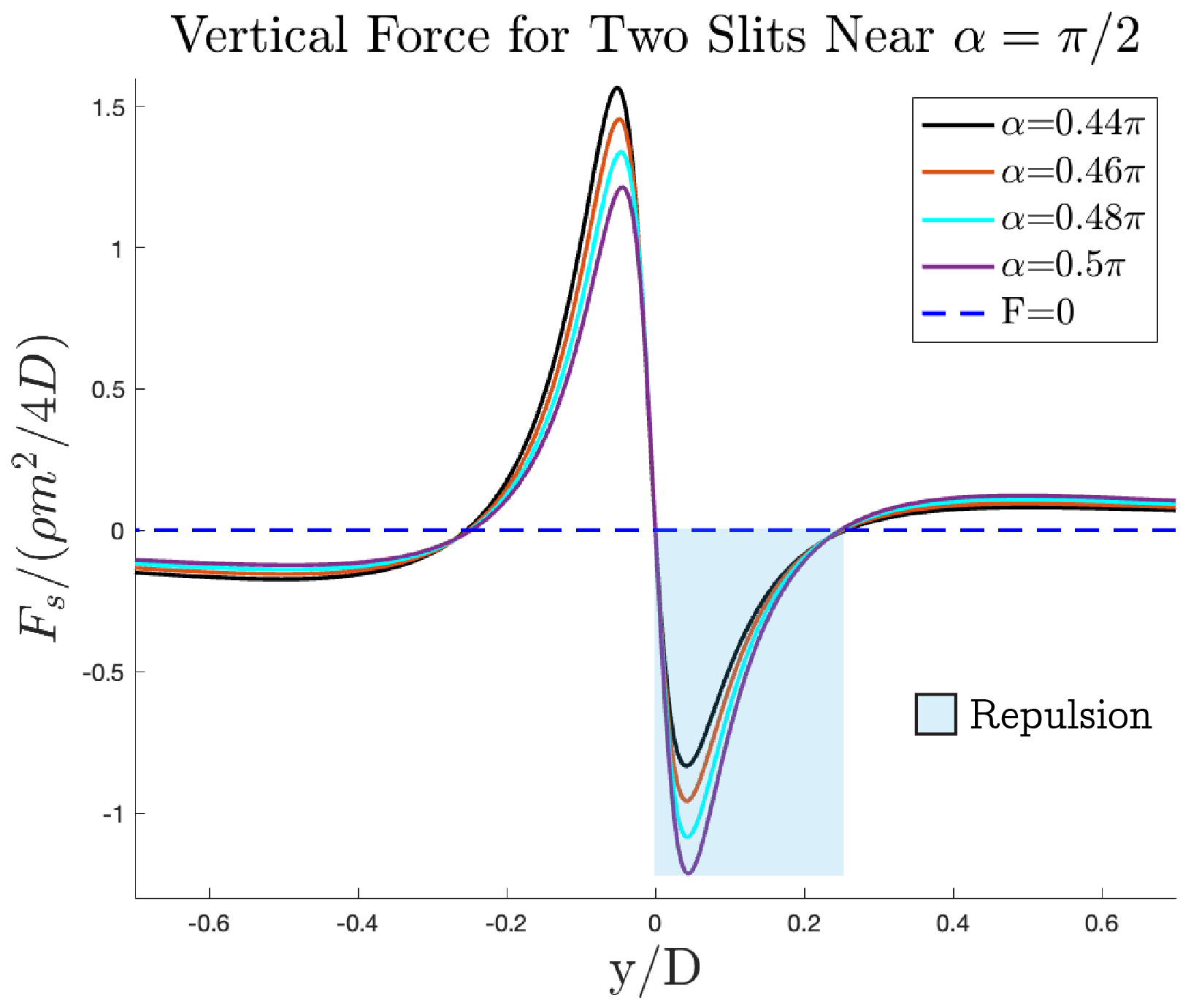}}
  \caption{Vertical force on the slits shown in figure \ref{fig:geometry}B, near $\alpha=\pi/2$ (slits on real axis). For large $|y/D|$ there is attraction as in the usual Bernoulli gripper. When $y=0$, slits lie on streamlines and there is local repulsion. For tilted slits, there is only repulsion for $y/D\gtrsim 0$ as indicated by the shaded region.}
\label{fig:plots2body}
\end{figure}

\section{Conclusion}
It follows from the minimum energy theorem of \citet[]{kelvin1849} that boundaries cannot reduce the kinetic energy and hence the added mass of a translating body. However, boundaries that resemble streamlines can create local energy minima as a function of the boundary-separation distance, resulting in non-intuitive effects on fluid forces. 

By a similar analysis in the theory of electrostatics, \cite{levin2011} predicted a maximum in the electromagnetic energy when conducting boundaries approximate potential contours, which manifested as a repulsive electrostatic force of induction. Below, we outline the analogous effects found in an ideal fluid, due to the energy minimum occurring when boundaries approximate streamlines.

In a system where \citet[\S 137]{Lamb1975} found a repulsive unsteady force, we show a transition to attraction when boundaries approximate streamlines. In two models of a Bernoulli gripper, where the steady force is typically attractive, we show that the steady force reverses sign when the boundary approximates a streamline. Effects are most prominent with slit-type boundaries, but are shown in \S \ref{singlebody} to persist for bodies with finite thicknesses. The geometry-dependent sign-reversal of the steady force in the Bernoulli gripper might be useful in future engineering design. Furthermore, similar effects may be relevant in electromagnetic power flows in 2D near-zero-index media, which were recently shown to be mathematically equivalent to ideal fluid flows \citep{liberal2020near}. \\
\\
\newline
\paragraph{\textbf{Acknowledgements.} The author is grateful for valuable discussions with Peter Baddoo, Bavand Keshavarz, Valeri Frumkin, Ousmane Kodio, and John Bush.}\\
\\
\paragraph{\textbf{Funding.} The author was funded by an MIT Presidential Fellowship during this work.}\\
\\
\paragraph{\textbf{Declaration of Interests.} } The authors report no conflict of interest.

\appendix
\section{Proof of Added Mass Increase by Conformal Maps}\label{altproof}
We prove that the added mass force may not decrease in the presence of a stationary boundary. The analysis follows that of \cite{Smith2008} and \cite{Tchieu2010}.
\subsection{Complex Variable Formulation}
Suppose a two-dimensional body $B$ translates in an ideal fluid with velocity $(V_x,V_y)$, which we write in complex notation as $V\equiv V_x+\mathrm{i} V_y$. Suppose also that there is a stationary boundary defined by $\partial D$, and fluid occupies the region between $\partial D$ and $\partial B$ (which may be finite or unbounded). The fluid flow between the boundaries is characterized by a complex potential satisfying the Laplace equation, $\nabla^2 W(z)=0$, where $z=x+\mathrm{i}y$. The velocity field is given by $U(x,y)=\overline{dW(z)}/dz=u(z)+\mathrm{i}v(z)$, where the bar denotes complex conjugation. The no-flux boundary condition takes the form $\Real\{\overline{U}n\}=\Real\{\overline{V}n\}$ for $z\in \partial B$, and $\Real\{\overline{U}n\}=0$ for $z\in \partial D$, where $n$ is the complex normal vector to the boundary.

 This problem is conveniently phrased using a conformal map. By the doubly-connected version of Riemann mapping theorem \citep[pp. 208]{goluzin}, there exists a conformal map between the annulus and a generic two-body geometry. Thus, the Laplace equation may be solved in the annulus, subject to appropriate boundary conditions, and taken to the physical plane using a conformal map \citep[]{Tchieu2010}. We take the annulus defined by $|\zeta|\in [\rho,1]$ without loss of generality. We also specify that the conformal map $z(\zeta)$ takes $|\zeta|=1$ to $\partial B$ and $|\zeta|=\rho$ to $\partial D$. The Neumann boundary conditions become
\begin{equation}\label{eq:Dirichlet_BC}
    \Imag\{W(\zeta)\}=\Imag\{\overline{V}z(\zeta)\},
\end{equation}
on each boundary; note that the second boundary is stationary so $V=0$ there. We can also take the velocity of $B$ to be purely imaginary, $V=Ui$ with $U\in\mathbb{R}$.
\begin{figure}
  \centerline{\includegraphics[height=8cm]{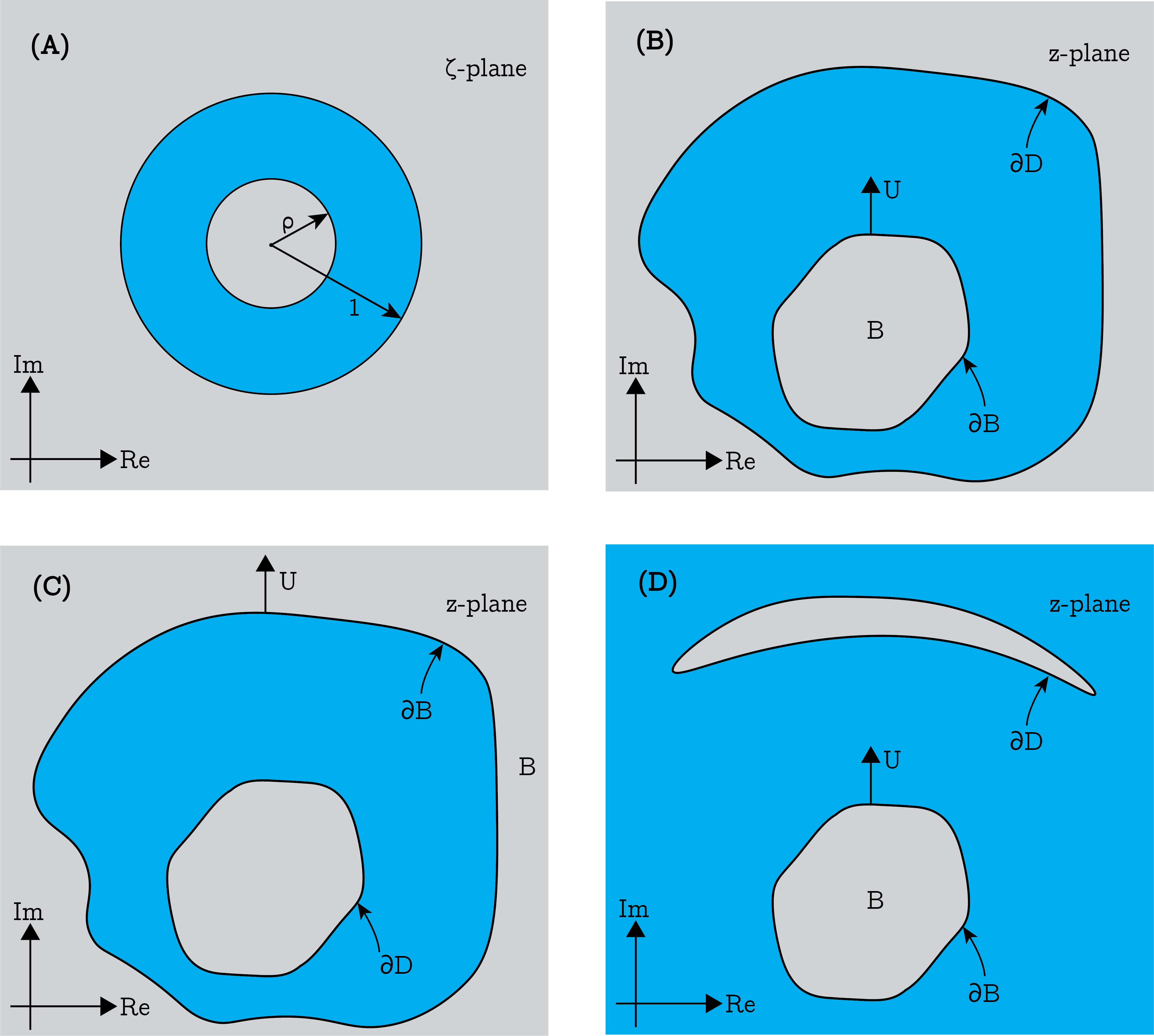}}
  \caption{Conformal Map Geometries. \textbf{(A)} The annulus geometry in the $\zeta$-plane. The conformal map $z(\zeta)$ maps the outer circle to $\partial B$ and the inner circle to $\partial D$. \textbf{(B)} Example of inverted confined geometry ($b=0$). \textbf{(C)} Example of a confined geometry without inversion. \textbf{(D)} Example of an unbounded fluid geometry ($b\neq 0$), the conventional two-body problem.}
\label{fig:confmaps}
\end{figure}
The conformal map may be written as a general Laurent series plus a simple pole in the fluid domain,
\begin{equation}\label{eq:CM_Laurent}
    z(\zeta)=\sum_{k\in \mathbb{Z}}a_k\zeta^k+\frac{b}{\zeta-c},
\end{equation}
where $c$ is located in the fluid domain, $|c|\in (\rho,1)$. For the confined problem, there is no fluid at infinity and thus no singularity, i.e., $b=0$. The area enclosed by $\partial B$, $A_1$, may be computed by Green's theorem though it never enters direct calculations.
\subsection{Force Calculation and Added Mass Increase}
The unsteady force on $B$ is given by the Blasius formula in the $\zeta$-plane,
\begin{equation}\label{eq:unsteadyforce}
    F_{US}=\mathrm{i}\frac{d}{dt}\left( \oint_{|\zeta|=1}{z(\zeta)\frac{dW(\zeta)}{d\zeta}d\zeta} \right)+A_1\ddot{z},
\end{equation}
where the fluid density has been set to 1 \citep[]{Tchieu2010}. The second term contributes negatively to the added mass, while the first has been seen in examples \citep{Tchieu2010} to contribute positively so the added mass is overall a positive quantity. Herein, we will analyse (\ref{eq:unsteadyforce}) and show that the boundary will never decreases the added mass relative the boundary-less problem. The transient problem with $U=0$ and $\dot{U}=1$ defines the added mass coefficient; equivalently, we solve for the force (\ref{eq:unsteadyforce}) by setting $U=1$ and eliminating the time derivative. The added mass coefficient opposing the acceleration is then defined (noting that acceleration is along the imaginary axis) by the relation,
\begin{equation}\label{eq:addedmass_co}
    C_a(\rho)=-\Imag\left\{\frac{F_{US}}{ A_1}\right\}=-\frac{2\pi}{A_1}\Imag\left\{{\mbox{Res}_{|\zeta|<1}} \left( z(\zeta) \frac{dW(\zeta)}{d\zeta} \right) \right\}-1,
\end{equation}
by the residue theorem. The residue term in (\ref{eq:addedmass_co}) may be expressed in terms of the coefficients of the mapping after imposing the boundary conditions. The boundary data is only non-zero on the unit circle $|\zeta|=1$, and is expanded there in a Laurent series as
\begin{equation}\label{eq:data_sing}
    \Imag\{\overline{V}z(\zeta)\}=-\frac{1}{2}\left(\sum_{n\in \mathbb{Z}}\left(a_n+\bar{a}_{-n} \right)\zeta^n+\frac{b}{\zeta-c}+\frac{\bar{b}\zeta}{1-\bar{c}\zeta}\right)=\sum_{k\in \mathbb{Z}}b_k \zeta^k,
\end{equation}
where we have set $V=\mathrm{i}$ as noted earlier. The residue theorem gives the data coefficients,
\begin{equation}\label{eq:b_coeffs}
  b_n = \left\{
    \begin{array}{ll}
      -\frac{1}{2}\left( a_n+\bar{a}_{-n} +\bar{b}\, \bar{c}^{\;n-1}\right) & n>0 \\
      -\frac{1}{2}\left( a_n+\bar{a}_{-n} +b c^{-n-1}\right) & n\leq0,
  \end{array} \right.
\end{equation}
noting that $|c|<1$. The complex potential is obtained by matching the data of (\ref{eq:b_coeffs}) onto $W(\zeta)$ at the boundaries. We expand in a Laurent series, $W(\zeta)=\sum_{k\in \mathbb{Z}}w_k\zeta^k$. As noted by \citet[]{Tchieu2010}, the imaginary part of $W(\zeta)$ can be written as an analytic function using the Schwartz conjugate on each boundary giving two equations,
\begin{equation}\label{eq:analyticw}
    \Imag\{W(\zeta)\}=\frac{1}{2i}\left( W(\zeta)-\bar{W}(\rho_i^2\zeta^{-1}) \right)=\frac{1}{2i}\sum_{k\in \mathbb{Z}}\left(w_k-\overline{w}_{-k}\rho_{i}^{-2k}\right)\zeta^k.
\end{equation}
where $\rho_i\in\{\rho,1\}$. The null boundary data on $|\zeta|=\rho$ relates the coefficients by $\rho^{2k}w_k=\bar{w}_{-j}$. Meanwhile, matching the $\rho_i=1$ equation of (\ref{eq:analyticw}) onto the data in (\ref{eq:b_coeffs}) yields $w_k=2ib_k/(1-\rho^{2k})$, which is explicitly
\begin{equation}\label{eq:w_coeffs}
    w_n=\left\{
    \begin{array}{ll}
      -\mathrm{i}\left( a_n+\bar{a}_{-n} +\bar{b}\, \bar{c}^{\;n-1}\right)/(1-\rho^{2n}), & n>0 \\
      -\mathrm{i}\left( a_n+\bar{a}_{-n}+bc^{-n-1}\right)/(1-\rho^{2n}) & n\leq0. \end{array} \right.
\end{equation}
The derivative of the potential is then obtained by term-by-term differentiation, $dW/d\zeta=\sum_{k\in\mathbb{Z}}kw_k\zeta^{k-1}$. Now $C_a(\rho)$ can be computed explicitly using the residue theorem, in terms of the coefficients of the map of (\ref{eq:CM_Laurent}); the explicit expression for $C_a$ is suppressed for brevity.
We now show that the effect of a boundary is never to decrease the added mass (\ref{eq:addedmass_co}), for a given shape $B$, relative to the boundary-less problem. Note that $\rho\rightarrow 0$, corresponds to the boundary-less problem \citep{crowdy2020solving}. The problem can be phrased as: \emph{can one choose coefficients $\{a_i,b\}$ in} (\ref{eq:CM_Laurent}) \emph{such that $C_a(\rho)/C_a(0)<1$ for some $\rho\in (0,1)$?} We begin by assuming $C_a(\rho)/C_a(0)<1$ holds and proceed toward a contradiction.

After letting $B\equiv bc^{k-1}$, some algebra reveals that $C_a(\rho)/C_a(0)<1$ is equivalent to
\begin{eqnarray}
 \label{eq:ineq_unbound1}
        \Real\left\{ \sum_{k>0} \frac{k \rho^{2k}}{1-\rho^{2k}} \left(B\bar{B} + 2 \bar{B} (\bar{a}_{k}+ a_{-k})  +(a_k+\bar{a}_{-k}) (a_{-k}+\bar{a}_k) \right) \right\}<0.
\end{eqnarray}
The left hand side may be bounded below by applying the Cauchy--Schwarz inequality. It follows that (\ref{eq:ineq_unbound1}) implies the weakened inequality
\begin{equation}\label{eq:ineq_unboundFinal}
         \sum_{k>0} \frac{k \rho^{2k}}{1-\rho^{2k}} \left(|B|^2 - 2 |B| |\bar{a}_{k}+ a_{-k}|  +|a_k+\bar{a}_{-k}|^2 \right)<0,
\end{equation}
which is impossible since the left side is positive semi-definite. Therefore, it is true that $C_a(\rho)\geq C_a(0)$.

\section{Kelvin's Theorem With Circulation}\label{App_kelvCirc}
The main assumption in the proof of Theorem \ref{thm1} was that the velocity potential $\phi$ was single-valued. This assumption implies that the corresponding irrotational flow $\boldsymbol{v}=\bnabla \phi$ has zero circulation around any closed loop. Here, we will demonstrate how Kelvin's energy theorem, as presented in Theorem \ref{thm1}, may be generalized in the presence of non-zero circulation. Note that a statement resembling the theorem to be proved here was given in words by \citet[\S 56]{Lamb1975}, but not proved there. 

We consider a flow in domain $D$ as described in Theorem \ref{thm1}, and we define the boundary domain $B\equiv \mathbb{R}^2\backslash D$. Suppose that $B$ is comprised of $N$ disconnected parts labelled $B_i$, each containing some circulation so that we can write the complex potential function as 
\begin{equation}\label{eq:pot_circ}
    W_{\Gamma}(z)=\tilde{W}(z)+\mathrm{i}\sum_{k=1}^N\frac{\Gamma_k}{2\pi} \log{\left(z-z_k\right)},
\end{equation}
where $z_k \in B_k$ and $\tilde{W}(z)$ is a single valued function over $D$. $\Gamma_k$ gives the circulation around the body $B_k$. The multivaluedness of $\phi=\mathrm{Re}\{W\}$ comes from the branch cuts of the logarithm terms. We now state and subsequently prove a generalization of Theorem \ref{thm1}, applicable to a general potential of the form given in (\ref{eq:pot_circ}).

\begin{theorem}\label{thm3}
Consider an ideal fluid domain $D\subseteq \mathbb{R}^2$, with velocity field $\boldsymbol{v}(\boldsymbol{x})=\bnabla \phi(\boldsymbol{x})$ satisfying $\nabla^2\phi=0$ with no-flux boundary conditions on $\partial D$. Further assume the boundary $B\equiv \mathbb{R}^2 \backslash D$ may be subdivided into $N$ disjoint regions $B_1,B_2,...,B_N$. Also assume $\phi=\mathrm{Re}\{W\}$ to single-valued up to some amount of circulation around each boundary such that the complex potential is written as,
\begin{equation}
        W(z)=\tilde{W}(z)+\mathrm{i}\sum_{k=1}^N\Gamma_k \log{\left(z-z_k\right)},
\end{equation}
for some circulations $\Gamma_k\in \mathbb{R}$ and where $z_k\in B_k$. If the fluid domain $D$ extends to infinity, it should be further assumed that $\sum_{k=1}^N \Gamma_k=0$, so that the energy intergal converges at infinity.
Further assume that the volume of each boundary does not change in time and the velocity vanishes at infinity. Then any other incompressible flow $\boldsymbol{q}(\boldsymbol{x})$, satisfying the boundary conditions on $\partial D$, possesses a kinetic energy greater than or equal to that of $\boldsymbol{v}$,
\begin{equation}\label{eq:theoremkelvin_circ}
     \frac{1}{2} \rho \int_{D}||\boldsymbol{v}||^2 dV \leq \frac{1}{2} \rho \int_{D}||\boldsymbol{q}||^2 dV,
\end{equation}
provided that
\begin{equation}\label{eq:theoremkelvin_circ_cond}
     \sum_{k=1}^{N-1}\left(\sum_{j=1}^{k}\Gamma_j\right)\int_{L_k}\left( \boldsymbol{q}- \boldsymbol{v}\right)\cdot \boldsymbol{n}dA\geq 0,
\end{equation}
where $L_k$ is a line connecting $z_k$ to $z_{{k+1}}$; $L_k$ are chosen not to intersect one another.
As in Theorem \ref{thm1}, we assume the bulk flow $||\boldsymbol{v}||^2$ contains no non-integrable singularities.
\end{theorem}

It is clear that in the case that $\Gamma_i=0$ for all $i$, Theorem \ref{thm1} is reproduced. An alternative condition for the applicability of (\ref{eq:theoremkelvin_circ}) is that the flux through each line $L_i$ is equal for the flows being compared, $\boldsymbol{q}$ and $\boldsymbol{v}$. In that case the integrals in (\ref{eq:theoremkelvin_circ_cond}) all vanish. This condition is reminiscent of that proposed by \cite{saad2017}, who generalized Kelvin's minimum theorem to simply connected domains with porous boundaries, while assuming a single-valued potential. We now prove Theorem \ref{thm3}.
\begin{proof}
The flow is decomposed according to $\boldsymbol{q}=\boldsymbol{v}+\boldsymbol{w}$, where $\boldsymbol{w}$ is finite in $D$. We then find that $(||\boldsymbol{q}||^2-||\boldsymbol{v}||^2)=2\bnabla \phi \cdot \boldsymbol{w}+||\boldsymbol{w}||^2$. The integrand is integrable and the product rule and divergence theorem apply, yielding
\begin{equation}\label{eq:theoremkelvin_proof_circ}
    \frac{1}{2} \rho \int_{D}\left(||\boldsymbol{q}||^2-||\boldsymbol{v}||^2\right) dV =    \rho\left(\int_{C}\phi \boldsymbol{w\cdot n}dA-\int_{D} \phi \bnabla \cdot \boldsymbol{w} dV\right)+\frac{1}{2} \rho \int_{D} ||\boldsymbol{w}||^2dV.
\end{equation}
The second term on the right hand side vanishes by the incompressibility condition on $\boldsymbol{q}$ and the last term is non-negative. It is left to show that the first term on the right hand side is greater or equal to zero. The contour $C$, over which that integral is taken, is simply the boundary $\partial D$ if $\phi$ is a single valued function. Otherwise, it must be chosen carefully so as to not cross branch cuts in order for the divergence theorem to be applicable (see figure \ref{fig:circulation}B). If the net sum of circulations vanishes, as we assumed, we can always choose the branch points to connect to one another in the following manner so that there are $N-1$ distinct branch cuts for $N$ disconnected boundaries. 

First, label each disconnected part of the boundary, $B_1,B_2, ... B_N$. Beginning with $B_1$, choose its branch cut, which begins at $z_1$, to pass through the branch point inside $B_2$; label this line $L_1$. The discontinuity across $L_1$ is $\Gamma_1$. Next, take both the branch cuts associated with $B_1$ and $B_2$ to extend to $z_3 \in B_3$. We choose the the branches of $B_1$ and $B_2$ to coincide on a line labelled $L_2$, which then has a discontinuity $\Gamma_1+\Gamma_2$. Repeating this procedure, we arrive at the line $L_{N-1}$ connecting $B_{N-1}$ to $B_N$, whose discontinuity is give by $\sum_{j=1}^{N-1}\Gamma_j$. Taking now all the branches to extend to complex infinity along the same line $L_N$, the associated discontinuity is $\sum_{j=1}^{N}\Gamma_j=0$, and the branch cut cancels. This procedure for choosing branch cuts is is illustrated for $N=3$ in figure \ref{fig:circulation}. The integral around $C$ in (\ref{eq:theoremkelvin_proof_circ}) vanishes by the no flux condition, except along the lines $L_1, L_2, ..., L_{N-1}$.

The contributions of those integrals come from the fact that the logarithm is $2\pi$-discontinuous across each branch cut. On the $L_1$ segment, the contribution to the integral is $\rho\Gamma_1\int_{L_1}\boldsymbol{w\cdot n}dA$. Physically, this is the flux of $\boldsymbol{w}=\boldsymbol{q}-\boldsymbol{v}$ through $L_1$ times the discontinuity across the branch cut. Similarly, the contribution from segment $L_k$ is given by the integral $\rho\left(\sum_{j=1}^{k-1}\Gamma_j\right)\int_{L_k}\boldsymbol{w\cdot n}dA$; physically, this represents the flux $\boldsymbol{w}$ through $L_k$ times the discontinuity across $L_k$. Summing over all segments, we find
\begin{equation}\label{eq:gamm_cond}
\rho\int_{C}\phi \boldsymbol{w\cdot n}dA=\rho\sum_{i=1}^{N-1}\left(\sum_{j=1}^{i-1}\Gamma_j\right)\int_{L_i}\boldsymbol{w\cdot n}dA.
\end{equation}
We require the right hand side to be greater or equal to zero, for (\ref{eq:theoremkelvin_circ}) to hold.

\begin{figure}
  \centerline{\includegraphics[height=5.3cm]{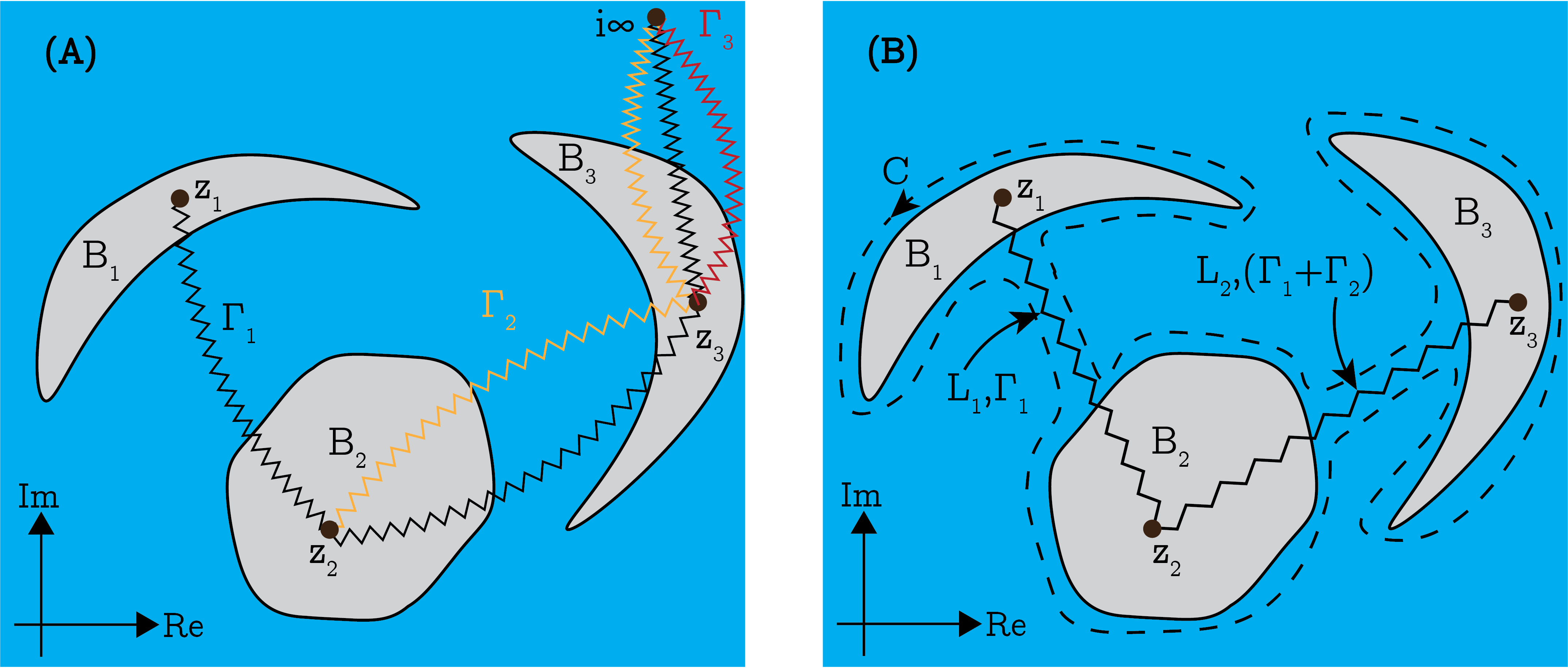}}
  \caption{Branch cut choices in the case $N=3$. In panel (A), we show a possible choice of branch cuts when each body $B_i$ possesses a circulation $\Gamma_i$. Branch cuts are drawn as zigzag lines, and labelled with its associated discontinuity. Each colour represents a single branch cut with a particular discontinuity. In panel (B), we take the branch cuts to overlap so that there are 2 distinct branches. The branch connecting to complex infinity is cancelled since $\Gamma_1+\Gamma_2+\Gamma_3=0$. The integration contour $C$ around the branch cuts is drawn.}
\label{fig:circulation}
\end{figure}

\end{proof}
\newpage
\bibliographystyle{jfm}
\bibliography{manuscript_text}

\end{document}